\documentclass[a4paper,12pt]{amsart}
\usepackage{amsmath}
\usepackage{amsthm}
\usepackage{amssymb}

\newcommand{\abs}[1]{\vert #1 \vert}  
\newcommand{\ov}{\overline}

\newcommand{\dif}{\mathrm{d}}

\newcommand{\HH}{\mathcal{H}}
\newcommand{\VV}{\mathcal{V}}

\newcommand{\RR}{\mathbb{R}}

\DeclareMathOperator{\di}{div}
\DeclareMathOperator{\gr}{grad}

\DeclareMathOperator{\Scal}{Scal}
\DeclareMathOperator{\Ric}{Ric}
\DeclareMathOperator{\Hess}{Hess}

\newtheorem{te}{Theorem}
\newtheorem{pr}{Proposition}
\newtheorem{co}{Corollary}
\newtheorem{lm}{Lemma}
\newtheorem{de}{Definition}    
\newtheorem{re}{Remark}

\begin{document}

\title[Shear-free perfect fluids with linear EoS]{Shear-free perfect fluids with linear equation of state}

\author{Radu Slobodeanu}

\address{Department of Theoretical Physics and Mathematics, Faculty of Physics, University of Bucharest, P.O. Box Mg-11, RO--077125 Bucharest-M\u agurele, Romania \ $\text{and}$}

\address{Institute of Mathematics, University of Neuch\^atel, 11 rue Emile Argand, 2000 Neuch\^atel, Switzerland.}

\email{radualexandru.slobodeanu@g.unibuc.ro}

\begin{abstract}
\noindent We prove that shear-free perfect fluid solutions of Einstein's field equations must be either expansion-free or non-rotating (as conjectured by Ellis and Treciokas) for all linear equations of state $p = w \rho$ except for $w \in \left\{-\tfrac{1}{5}, -\tfrac{1}{6}, -\tfrac{1}{11}, -\tfrac{1}{21}, \tfrac{1}{15}, \tfrac{1}{4}\right\}$. 
\end{abstract}

\thanks{I thank to E. Loubeau for many helpful discussions.}

\subjclass[2010] {53B30, 53C43, 53Z05, 83C55.}

\keywords{Harmonic morphism, perfect fluid, Einstein equations.}

\maketitle

\section{Introduction}
For the Einstein field equations with perfect fluid sources, an important insight into the structure of solutions can be obtained under two fairly general and observationally defensible assumptions: the equation of state for the pressure $p$ and energy density $\rho$ is barotropic (i.e. $p=p(\rho)$) and the flow is shear-free (i.e. fluid's velocity is a transversally conformal vector). While the former restriction is specific to isoentropic fluids, the latter is essentially a kinematic condition of isotropy. In cosmology, the shear-free condition expresses the isotropy of the relative recessional motion of the galaxies (but allows the red shift and the cosmic microwave background radiation to be anisotropic) and is a common feature of standard spacetimes, such as the FRW and G\"odel models. In kinetic theory this characterizes \cite{te} the velocity of collision-dominated gases with isotropic distribution function, under the Einstein-Boltzmann equations.  

The first result indicating the nature of solutions in this context was stated by G\"odel (\cite{god}): a shear-free dust fluid (i.e. $p=0$, in particular a congruence of timelike geodesics) in a spatially homogeneous spacetime of type IX cannot both expand and rotate. This contrasts Newtonian cosmology where expanding and rotating (shear-free) solutions exist and can avoid the singularity formation \cite{nar}. G\"odel's result was later  proved \cite{e} to remain true without symmetry assumptions and it can be seen as a timelike analogue of the Goldberg-Sachs theorem \cite{gold}. After proving that the same conclusion holds for radiation fluids (i.e. $p=\tfrac{1}{3}\rho$), Ellis and Treciokas (\cite{te}) formulated the \emph{shear-free fluid conjecture}:
\begin{quote} 
If the velocity vector field of a self-gravitating barotropic perfect fluid ($p + \rho \neq 0$ and $p = p(\rho)$) is shear-free, then either the expansion or the rotation of the fluid vanishes.
\end{quote}

Since then, many partial results have been obtained by specializing $(i)$ the curvature (divergence-free electric or magnetic part of Weyl curvature \cite{carmi, ber1}, Petrov types N \cite{carm1} and III \cite{carm2}), $(ii)$ kinematical quantities (parallel vorticity and acceleration \cite{coll}, functionally dependent expansion and energy density \cite{lancol} or expansion and rotation scalar \cite{sop}, fluid flow parallel with a conformal vector field \cite{coley}), or  $(iii)$ the equation of state ($p=w\rho$, with $w=0, \pm \tfrac{1}{3}, \tfrac{1}{9}$ \cite{e, la, slob, te, ber0}). Further aspects supporting the conjecture can be found in \cite{col}; for a recent discussion, see \cite{el}.

In Riemannian geometry, several results of similar flavour have been independently discovered. They are all related to the notion of ($r$-) harmonic morphism (see \cite{ud} for a general account). To connect them to the above results in general relativity, let us recall that the relativistic perfect fluids (whether self-gravitating or not) can benefit from a variational treatment if we define the action functional in terms of a submersion from the (4-dimensional) spacetime having its fibres tangent to the fluid velocity, $U$. This approach, initiated in \cite{car, langl, tau}, has been mathematically settled in \cite{cristo} and renewed in recent physics literature, e.g. \cite{ratta, nico}. In particular, it turns out \cite{slo} that a shear-free relativistic perfect fluid with linear equation of state $(U, p = \frac{r - 3}{3}\rho, \rho = \lambda^{r})$ ($r\neq 0$) is determined by a (locally defined) horizontally conformal submersion which is a critical point of the action functional, i.e. an $r$-harmonic morphism (\cite{lube}) with 1-dimensional timelike fibres tangent to $U$ and dilation $\lambda$. In Riemannian signature and dimension at least four, precisely in the 1-dimensional fibres case, 2-harmonic morphisms on Einstein manifolds and $(r\geq 2)$-harmonic morphisms on constant curvature spaces have been classified (\cite{bry, mo, pan, pant}): they are all either of \textit{Killing type} (with $U$ collinear with a Killing vector field and $\di U=0$, so "expansion-free"), of \textit{warped product type} (with $U$ self-parallel and orthogonal to a foliation by hypersurfaces, so "non-rotating"), or (only on Einstein four-manifolds) of a \textit{third type} which implies Ricci flatness. A similar result holds also on conformally-flat manifolds \cite{panti}. 

Adopting the standpoint of $r$-harmonic morphisms, in this paper we prove the shear-free fluid conjecture for all linear equations of state $p = w \rho$ except for six values of $w$ which require a case by case analysis. Together with the four values already settled, they form the full set of exceptional values of the present proof. 

\section{Preliminaries}
In this section we present the terminology related to perfect fluids and introduce some basic results needed latter on. For a summary of notations and conventions, see the Appendix.

\subsection{Perfect fluids coupled with gravity}
\begin{de} $($\cite{mis,one}$)$
Let $(M,g)$ be a $4$-dimensional spacetime. A triple $(U, p, \rho)$ is called (relativistic) \emph{perfect fluid} if:
\begin{enumerate}
\item[$(i)$] $U$ is a timelike future-pointing unit vector field on $M$, called the \emph{flow vector field} (or \emph{normalized 4-velocity}),

\item[$(ii)$] $\rho, p: M\to \RR$ are real functions, called \emph{mass (energy) density} and \emph{pressure}, respectively,

\item[$(iii)$] the \emph{stress-energy tensor} of the fluid is conserved: 
$$\di \left(p \, g + (p + \rho)\varpi \otimes \varpi\right)=0,$$
where $\varpi(X) = g(U, X)$, for all $X$ tangent to $M$.
\end{enumerate}
If instead of $(iii)$, the \emph{Einstein field equations} are satisfied
\begin{equation}\label{einstein}
\Ric - \tfrac{1}{2}\Scal \cdot g = p \, g + (p + \rho)\varpi \otimes \varpi,
\end{equation} 
then  $(U, p, \rho)$ is called \emph{perfect fluid coupled to gravity} (or \emph{self-gravitating}). 
\end{de}

Condition $(iii)$ of the definition splits into the fluid's equations:
\begin{equation}\label{euler}
\begin{split}
(\rho+p)\nabla_{U}U + \gr^{\HH}p &=0 \quad (\text{Euler equations}),\\
(\rho + p)\di U + U(\rho) &=0  \quad (\text{energy conservation}),
\end{split}
\end{equation}
where $\gr^{\HH}p$ is the \emph{spatial pressure gradient}.
The coupling with gravity imposes the following block-diagonal form of the Ricci tensor:
\begin{equation}\label{couple}
\begin{split}
\Ric(U,U) &= \tfrac{1}{2}(\rho + 3p),  \quad \Ric(X,U) = 0, \quad \forall X \perp U,\\
\Ric(X,Y) &= \tfrac{1}{2}(\rho - p)g(X,Y) , \quad \forall X,Y \perp U.
\end{split}
\end{equation}
Notice that $\Scal = \rho-3p$.

We denote by $\VV$ the (vertical) foliation spanned by $U$ and by $\HH$ its (horizontal) complement in the tangent bundle of $M$. The orthogonal projection on $\HH$ (or $\VV$) of a vector field $X$ will be denoted by $X^\HH$ (or $X^\VV$). Then any vector $X$ has the splitting: $X = X^\HH - \varpi(X)U$. Let $\nabla U$ be the tensor defined by $(\nabla U)(X,Y)=g(\nabla_X U, Y)$, for all $X$, $Y$.

\begin{pr} The covariant derivative of the flow vector field decomposes as follows
\begin{equation}\label{deco}
\nabla U + \varpi \otimes (\nabla_U U)^\flat = \omega + \sigma +
\tfrac{1}{3}\di U \, g^\HH ,
\end{equation}
where $g^\HH(X, Y)=g (X^\HH, Y^\HH)$ is the \emph{orthogonal projector} to $U$ and
\begin{equation}\label{defkin}
\begin{split}
\omega(X, Y)&=\tfrac{1}{2} \dif \varpi (X^\HH, Y^\HH),\\
\sigma(X, Y)&=\tfrac{1}{2} \left( \mathcal{L}_U g -
\tfrac{2}{3} \di U \, g \right)(X^\HH, Y^\HH).
\end{split}
\end{equation}
\end{pr}

\begin{proof}
Notice that Koszul formula can be rewritten as follows:
$$
g(\nabla_X U, Y)=\dif \varpi (X,Y)+\tfrac{1}{2}(\mathcal{L}_U g)(X, Y), \quad \forall X,Y.
$$
Since $g(\nabla_X U, Y) = g(\nabla_{X^\HH} U, Y^\HH) - \varpi(X)g(\nabla_U U, Y)$, applying Koszul formula on the first right-hand term yields the result.
\end{proof}

\begin{de}[Kinematical quantities] For a relativistic fluid $(U, p, \rho)$, $\nabla_U U$ is the \emph{acceleration} (vector), $\di U$ is the \emph{expansion rate} (scalar), $\omega$ is the \emph{rotation rate} or \emph{vorticity} (2-form) and $\sigma$ is the \emph{shear rate} (symmetric and traceless 2-covariant tensor).
\end{de}

From the following slightly reformulated Yano's identity (\cite{yan})
\begin{equation}\label{yano}
\mathrm{Ric}(X, Y)=\di\left(\nabla_X Y \right)- X(\di Y) - 
\tfrac{1}{2}\langle \mathcal{L}_{X}g , \mathcal{L}_{Y}g \rangle + \tfrac{1}{2}\langle\dif X^{\flat}, \dif Y^{\flat}\rangle ,
\end{equation}
where $X,Y$ are tangent vectors to a (semi-)Riemannian manifold, we obtain

\begin{pr} [Raychaudhuri equation \cite{ray}] On the spacetime $(M,g)$ the following identity hold:
\begin{equation}\label{ray}
U(\di U) = -\tfrac{1}{3}(\di U)^2 + 2\left(\abs{\omega}^2- \abs{\sigma}^2\right) + \di\left(\nabla_U U \right) - \mathrm{Ric}(U, U).
\end{equation}
\end{pr}

In this context, the shear-free fluid conjecture reads: if Equation \eqref{einstein} is satisfied by a fluid with $p=p(\rho)$, $p+\rho \neq 0$ and $\sigma=0$, then $(\di U)\omega=0$. We shall treat the special case described below.
\subsection{Shear-free linear-barotropic fluids} Now consider a shear-free ($\sigma=0$) perfect fluid with \textit{linear equation of state} $p=w\rho$ (which is the $\gamma$-\textit{law} used in cosmology). We can always write it in the form $(U, p = \frac{r - 3}{3}\lambda^{r}, \rho = \lambda^{r})$, where $\lambda$ is some positive function (locally) on $M$ ($\lambda^{-1}$ represents the \textit{length scale}) and $r \in \RR^*$. 

The fluid equations \eqref{euler} become, respectively
\begin{equation}\label{hm1}
- \nabla_U U + (3-r) \gr^{\HH}(\ln \lambda)=0,
\end{equation}
\begin{equation}\label{hm2}
\di U + 3U(\ln \lambda)=0.
\end{equation}
\begin{re} \label{dualit} According to \cite{bry, mo}, \eqref{hm1} and \eqref{hm2} are the necessary and sufficient conditions for $U$ to be tangent to the fibres of a (locally defined) $r$-harmonic morphism $\varphi: (M,g) \to (N, h)$ into some 3-dimensional Riemannian manifold, with dilation $\lambda$. Recall that \cite{lube} an $r$-harmonic morphism $\varphi$ is characterized as a horizontally conformal (i.e. $\varphi^* h  = \lambda^2 g^\HH$) critical point of the action functional $\int_M \abs{\dif \varphi}^r \mathrm{vol}_g$. The corresponding Euler-Lagrange equations are given precisely by \eqref{hm1}, while \eqref{hm2} is an identity fulfilled by any horizontally conformal submersion (for more details see \cite{slo}).
\end{re}

In this case, the Conditions \eqref{couple} for the fluid to be coupled with gravity turn out to be
\begin{equation}\label{couple+}
\begin{split}
\Ric(U,U) &= \tfrac{r-2}{2}\lambda^r , \quad \Ric(X,U) = 0 , \quad \forall X \perp U ,\\
\Ric(X,Y) &= \tfrac{6-r}{6}\lambda^r g(X,Y), \quad \forall X,Y \perp U. 
\end{split}
\end{equation}
and therefore $\Scal = (4-r)\lambda^r$.

A crucial notion in the sequel are basic tensors.
\begin{de} $($\cite{ud}$)$ A section $\varsigma$ of $(\otimes^r \HH) \otimes (\otimes^s \HH^*)$ is called \emph{basic} if $(\mathcal{L}_W \varsigma)^\HH =0$ for all $W \in \Gamma(\VV)$.
\end{de}
In particular, a function $f$ on $M$ is basic if $W(f)=0$, and a horizontal vector field $X$ on $M$  is basic if $[W, X]^\HH=0$, for all $W \in \Gamma(\VV)$;  the latter condition means that $X$ is \textit{projectable} on $N$.

Let us define the \textit{fundamental vector field} of $\VV$ by 
$$V = \lambda^{3-r} U.$$
Analogously to \cite{pan} (see also \cite[p.341]{ud}), we have
\begin{pr} The Euler equation \eqref{hm1} is equivalent to
\begin{equation}\label{basic}
[V, X]=0, 
\end{equation}
for all basic vector fields $X$ on $M$.
\end{pr}
\begin{proof}
Let $X$ be a basic vector field. Then $[V, X]$ is vertical. Since
$$
g([V, X], U)=\lambda^{3-r}g\big((3-r)\gr^\HH \ln \lambda - \nabla_U U , X\big),
$$
we conclude that $[V, X]=0$ if and only if Euler equation is satisfied.
\end{proof}

Let $\vartheta$ be the 1-form dual to $V$, defined by $\vartheta(X)=-\lambda^{2(r-3)}g(X, V)$ for all $X$, and let $\Omega = \dif \vartheta$.

\begin{lm}\label{basic+}
Let $(U, p = \frac{r - 3}{3}\lambda^{r}, \rho = \lambda^{r})$ be a shear-free perfect fluid. Then

$(i)$ $\Omega(X,Y)=\lambda^{2r-6}g([X,Y],V)$ for all $X, Y \in \Gamma(\HH)$;

$(ii)$ $\Omega=-2\lambda^{r-3}\omega$ ($\Omega$ is \emph{proportional to the rotation rate}); 

$(iii)$ $\imath_W \Omega=0$, for all $W \in \Gamma(\VV)$ ($\Omega$ is a \emph{horizontal} $2$-form) 

$(iv)$ $\Omega=0$ if and only if $\HH$ is integrable \footnote{In other words, for the 4-velocity of a perfect fluid, being hypersurface orthogonal is equivalent with being vorticity-free (i.e. having zero rotation).}; 

$(v)$ $\mathcal{L}_W \Omega=0$ or, equivalently, $W(\Omega(X,Y))=0$, for all $X, Y$ basic vectors and for all $W \in  \Gamma(\VV)$ ($\Omega$ is a \emph{basic} $2$-form).

$(vi)$ \  For any basic vector field $X$, the following function is basic:
$$\lambda^{-2}[\delta \Omega(X)+(r-4)\Omega(X, \gr \ln \lambda)].$$ 
\end{lm}

\begin{proof}
$(i)$ and $(ii)$ are immediate using only the definitions. For $(iii)$ we can check that $\Omega(V,X)=0$ for a basic vector $X$ by employing Equation \eqref{basic}. $(iv)$ is the Frobenius' theorem upon applying $(iii)$. $(v)$  is the consequence of Cartan's formula 
$\mathcal{L}_W \Omega=\imath_W\dif \Omega +\dif(\imath_W \Omega)$. 

To prove the last assertion $(vi)$, let $\{X, Y, Z\}$ be orthogonal basic vectors from a preferred frame (see Remark below). By direct computation we obtain
\begin{equation*}
\begin{split}
\delta\Omega(X)&=\lambda^2 \left[ (4-r) Y(\ln \lambda)\Omega(X, Y) + Y(\Omega(X, Y))
+ \lambda^2 g([Z,Y], Z)\Omega(X, Y) \right],\\
\delta\Omega(Y)&=\lambda^2 \left[(4-r) X(\ln \lambda)\Omega(Y, X) - X(\Omega(X, Y))
- \lambda^2 g([Z,X], Z)\Omega(X, Y) \right],\\
\delta\Omega(Z)&=-\lambda^4 g([X,Y], Z)\Omega(X, Y).
\end{split}
\end{equation*}
Since $[X,Y]^\HH$, $[X,Z]$ and $[Y,Z]$ are basic vectors, the functions $\lambda^2 g([X,Y], Z)$, $\lambda^2 g([Z,X], Z)$ and $\lambda^2 g([Z,Y], Z)$ are basic. Using also the fact that $\Omega$ is basic cf. $(v)$, we obtain the conclusion for $X$, $Y$ and $Z$, and therefore for any basic vector.
\end{proof}

\begin{re} [Preferred orthonormal frames, cf. \cite{ud, pan}] \  We can always choose a local orthogonal frame $\{X, Y, Z\}$ of basic horizontal vector fields and we may suppose that their lengths satisfy 
$\abs{X} = \abs{Y} = \abs{Z} = 1/\lambda$ 
and that $Z$ satisfies $\imath_Z\Omega = 0$ 
\footnote{Since $\Omega$ is basic, it locally descends to a 2-form on $N$ that can be seen as a skew-symmetric linear map $TN \to TN$. Since $\dim N$ is odd, such a map is singular.} $($i.e. $Z$ is collinear to the \emph{vorticity vector} $(\ast_{\HH} \omega)^\sharp )$. 
Then the contractions $\imath_X\Omega$ and $\imath_Y\Omega$ are both basic and orthogonal, i.e. $\langle \imath_X\Omega , \imath_Y\Omega \rangle =0$. Since $\dif \Omega(X,Y,Z)=0$, we also have
\begin{equation}\label{lz}
(\mathcal{L}_Z \Omega)(X,Y)=0.
\end{equation}
Moreover, $\abs{\Omega}^2=\lambda^4 \Omega(X,Y)^2$ and
$\abs{\imath_X\Omega}^2=\abs{\imath_Y\Omega}^2=\lambda^2 \Omega(X,Y)^2$.

\noindent We call $\{U, \lambda X, \lambda Y, \lambda Z\}$ a \emph{preferred orthonormal frame}
\footnote{A similar choice of an orthonormal tetrad was done in \cite{ber0}.}.
\end{re}

\section{Constraint equations}

The very existence of a shear-free perfect fluid constrains the spacetime geometry, in particular the Ricci curvature (Proposition \ref{ricci}) that, in addition, must have block-diagonal form \eqref{couple} by the Einstein field equations.
For a shear-free fluid with linear equation of state $p = \tfrac{r - 3}{3}\rho$, with $\rho = \lambda^{r}$ and $r \neq 2$, we see these curvature restrictions as constraints on the second order derivatives of $\ln \lambda$. This will provide compatibility conditions at the level of $3^{rd}$ order derivatives, in the form of polynomial relations in first derivatives only.

In this section we collect the constraint equations that are useful in the proof of the conjecture. Analogously to \cite{pan} (see also \cite[p.343]{ud}) we have
 
\begin{pr}[Ricci curvature restrictions]\label{ricci}
Let $(M,g)$ be a four dimensional spacetime and $\varphi : (M, g) \to (N, h)$ an $r$-harmonic morphism of dilation $\lambda$, into a $3$-dimensional Riemannian manifold. If the fibres are tangent to the timelike unit vector field $U$, then the following identities hold for all horizontal vectors $X,Y$:
\begin{equation}\label{ricphm1}
\Ric(U,U) = (3-r) \Delta \ln \lambda 
+ (6-r) U(U(\ln \lambda)) 
+ 3(r-4) U(\ln \lambda)^2 + \tfrac{\lambda^{2(3 - r)}}{2} \abs{\Omega}^2 , 
\end{equation}
\begin{equation}\label{ricphm2}
\Ric(X,U) =  2 X(U(\ln \lambda)) 
- \tfrac{\lambda^{3-r}}{2} \{ \delta \Omega (X) 
+ 2(3-r)\Omega (X, \gr \ln \lambda)\} ,
\end{equation}
\begin{equation}\label{ricphm3}
\begin{split}
\Ric(X,Y) &= \varphi^* \Ric^N (X, Y) 
+(r-2)\Hess_{\ln \lambda} (X,Y) \\
& - [(r-3)^2 + 1] X (\ln \lambda) Y(\ln \lambda) + \tfrac{\lambda^{2(3-r)}}{2} \langle \imath_X \Omega , \imath_Y \Omega \rangle \\
&+ \big(\Delta \ln \lambda +(r-2)\abs{\gr \ln \lambda}^2 \big)  g(X,Y).
\end{split}
\end{equation}
\end{pr}

\begin{proof}
The first identity is simply another way of writing Raychaudhuri equation \eqref{ray} in the shear-free case, by using Equations \eqref{hm1} and \eqref{hm2}.

The second identity can be derived directly from Yano's formula \eqref{yano}, taking into account the decomposition \eqref{deco}. This formula is known in physics literature as shear-divergence identity or $(0, \alpha)-$equation and it is one of the standard constraint equations, cf. e.g.  \cite{seno}.

The third identity is obtained by taking the trace of the following formula, true for any horizontally conformal map $\varphi:(M,g) \to (N,h)$ and for all horizontal vectors $X, Y, Z, T$  (\cite[p.320]{ud})
\begin{equation*}
\begin{split}
& \langle R(X, Y)Z, T \rangle = \lambda^{-2}\langle R^N(\dif \varphi(X), \dif \varphi(Y))\dif \varphi(Z), \dif \varphi(T))\rangle \\
&- \langle X (\ln \lambda) Y - Y (\ln \lambda) X, T (\ln \lambda) Z - Z(\ln \lambda) T \rangle + \{ \langle Y, Z \rangle \Hess_{\ln \lambda} (X, T) \\
& - \langle X, Z \rangle \Hess_{\ln \lambda} (Y, T) 
+ \langle X, T \rangle \Hess_{\ln \lambda} (Y, Z) 
- \langle Y, T \rangle \Hess_{\ln \lambda} (X, Z) \}\\
& + \tfrac{1}{4} \{\langle I(X, Z), I(Y,T) \rangle - \langle I(Y, Z), I(X,T) \rangle + 2\langle I(X,Y), I(Z,T)\rangle \}\\
&+ (\langle Y,Z \rangle \langle X,T \rangle - \langle X,Z \rangle \langle Y,T \rangle) \abs{\gr \ln \lambda}^2, 
\end{split}
\end{equation*}
where $I(X, Y)=[X, Y]^\VV$, combined with the \textit{propagation of shear}:
\begin{equation*}
\begin{split}
&R(U,X,U,Y)+\tfrac{1}{2}(\mathcal{L}_{\nabla_U U}g)(X,Y)-\tfrac{1}{3}\left[U(\di U)+ \tfrac{1}{3}(\di U)^2\right]g(X,Y) \\
&+ \langle \imath_{X} \omega, \imath_{Y} \omega \rangle + g(\nabla_U U, X)g(\nabla_U U, Y)=0, \quad \forall X,Y \in \HH.
\end{split}
\end{equation*}
\end{proof}

Using the gravity coupling condition \eqref{couple+} and Remark \ref{dualit}, we obtain
\begin{co}
Let $(U, p = \frac{r - 3}{3}\lambda^{r}, \rho = \lambda^{r})$ be a shear-free perfect fluid coupled with gravity on $(M,g)$. Then the identities \eqref{ricphm1}, \eqref{ricphm2} and \eqref{ricphm3} hold with the left-hand side replaced by $\frac{r-2}{2}\lambda^r$, $0$ and $\frac{6-r}{6}\lambda^r g(X,Y)$, respectively, where $N$ is endowed with a metric $h$ such that the projection along $U$, $\varphi: (M, g) \to (N, h)$, satisfies $\varphi^*h=\lambda^2 g^{\HH}$.
\end{co}

Taking the (horizontal) trace of \eqref{ricphm3} and combining with \eqref{ricphm1} gives:

\begin{co}[Trace constraints] On a spacetime $(M,g)$ with a shear-free perfect fluid 
$(U, p = \frac{r - 3}{3}\lambda^{r}, \rho = \lambda^{r})$ coupled  with gravity, the following identities must hold
\begin{equation}\label{tr1}
\begin{split}
\Delta \ln \lambda = &  \tfrac{r-2}{2} U(\ln \lambda)^2 
- \tfrac{r-4}{3}\lambda^{r} + \tfrac{3r-14}{24}\lambda^{6-2r}\abs{\Omega}^2\\
&+ \tfrac{(2r-5)(r-6)}{6}\abs{\gr^\HH \ln \lambda}^2 + \tfrac{r-6}{12}\lambda^{2}\Scal^N ,
\end{split}
\end{equation}
\begin{equation}\label{tr2}
\begin{split}
U(U( \ln \lambda)) = &  -\tfrac{r-5}{2} U(\ln \lambda)^2 
+ \tfrac{2r-5}{6}\lambda^{r} - \tfrac{3r-5}{24}\lambda^{6-2r}\abs{\Omega}^2\\
& - \tfrac{(2r-5)(r-3)}{6}\abs{\gr^\HH \ln \lambda}^2 - \tfrac{r-3}{12}\lambda^{2}\Scal^N .
\end{split}
\end{equation}
\end{co}

\subsection{Constraints in a preferred frame}

Let $\{\lambda X, \lambda Y, \lambda Z, U\}$ be a preferred orthonormal frame. In the following we will obtain some constraints for the existence of a shear-free fluid $(U, p = \frac{r - 3}{3}\lambda^{r}, \rho = \lambda^{r})$ coupled with gravity in terms of these basic vector fields $X$, $Y$ and $Z$.

Condition \eqref{ricphm2} is equivalent to
\begin{equation}\label{ricphm2+}
\begin{split}
X(V(\ln \lambda)) = & 
\tfrac{\lambda^{2(4-r)}}{4} \{ \beta(X) + (10- 3r)Y(\ln \lambda)\Omega (X,Y)\} + (3-r)X(\ln \lambda)V(\ln \lambda),\\
Y(V(\ln \lambda)) = & 
\tfrac{\lambda^{2(4-r)}}{4} \{ \beta(Y) - (10- 3r)X(\ln \lambda)\Omega (X, Y)\} + (3-r)Y(\ln \lambda)V(\ln \lambda),\\
Z(V(\ln \lambda)) = & 
\tfrac{\lambda^{2(4-r)}}{4} \beta(Z) + (3-r)Z(\ln \lambda)V(\ln \lambda),\\
\end{split}
\end{equation}
where $\beta(T):=\lambda^{-2}[\delta \Omega (T) 
+ (r-4)\Omega (T, \gr \ln \lambda)]$ is basic whenever $T$ is a basic vector, according to Lemma \ref{basic+}. 

Assume $r \neq 2$ and let $A(r)=\tfrac{(r-3)^2 + 1}{r-2}$ and $B(r)=1+\tfrac{(2r-5)(r-6)}{6(r-2)}$. Condition \eqref{ricphm3} combined with the trace constraints gives:
\begin{equation}\label{hessxx}
\begin{split}
\Hess_{\ln \lambda}(X,X)= &(A(r)-B(r))X(\ln \lambda)^2 - B(r)[Y(\ln \lambda)^2 + Z(\ln \lambda)^2] + \tfrac{1}{2}\lambda^{-2}U(\ln \lambda)^2\\
& +  \tfrac{1}{6}\lambda^{r-2} -  \tfrac{3r-2}{24(r-2)}\lambda^{8-2r}\Omega(X, Y)^2 \\
&-  \tfrac{1}{r-2} \left(\varphi^*\Ric^N (X,X) + \tfrac{r-6}{12}\Scal^N \right),
\end{split}
\end{equation}
the analogous equation for $\Hess_{\ln \lambda}(Y,Y)$ and
\begin{equation}\label{hesszz}
\begin{split}
\Hess_{\ln \lambda}(Z,Z)= &(A(r)-B(r))Z(\ln \lambda)^2 - B(r)[X(\ln \lambda)^2 + Y(\ln \lambda)^2] + \tfrac{1}{2}\lambda^{-2}U(\ln \lambda)^2\\
& +  \tfrac{1}{6}\lambda^{r-2} -  \tfrac{3r-14}{24(r-2)}\lambda^{8-2r}\Omega(X, Y)^2 \\
&-  \tfrac{1}{r-2} \left(\varphi^*\Ric^N (Z,Z) + \tfrac{r-6}{12}\Scal^N\right).
\end{split}
\end{equation}

Condition \eqref{ricphm3} on pairs of orthogonal vectors gives:
\begin{equation}\label{hessxy}
\Hess_{\ln \lambda} (X,Y) - A(r) X(\ln \lambda)Y(\ln \lambda) + \tfrac{1}{r-2}\varphi^* \Ric^N (X, Y)=0
\end{equation}
and similar equations for $(X,Z)$ and $(Y,Z)$.

\medskip
The previous equations prescribe second order derivatives of $\ln \lambda$ in terms of its first derivatives. Differentiating them  along $V$ (i.e. "propagating") and using commutation \eqref{basic} will provide us with compatibility conditions purely in terms of first derivatives.

By taking the derivative along $V$ of the equation (cf. \eqref{ricphm3})
\begin{equation}\label{hessxxyy}
\begin{split}
& \Hess_{\ln \lambda} (X,X) - A(r) X (\ln \lambda)^2 + \tfrac{1}{r-2}\varphi^* \Ric^N (X, X) \\
&=\Hess_{\ln \lambda} (Y,Y) - A(r) Y (\ln \lambda)^2 + \tfrac{1}{r-2}\varphi^* \Ric^N (Y, Y),
\end{split}
\end{equation}
inserting $X(V(\ln \lambda))$, $Y(V(\ln \lambda))$ from \eqref{ricphm2+} and $X(Y(\ln \lambda))$ from \eqref{hessxy}, and simplifying the result using again \eqref{hessxxyy}, we obtain:
\begin{equation}\label{vhessxxyy}
\begin{split}
&(3-r)V(\ln \lambda)\big\{(5-r-A(r))[X(\ln \lambda)^2 - Y(\ln \lambda)^2] \\
&+ \tfrac{1}{r-2}[\varphi^* \Ric^N (Y, Y)- \varphi^* \Ric^N  (X, X)]\big\} \\
=&-\tfrac{\lambda^{8-2r}}{4}\big\{2(10-3r)(13-3r-A(r))\Omega(X,Y) X(\ln \lambda) Y(\ln \lambda)\\
& + [(15-3r-2A(r))\beta(X) +(10-3r)Y (\Omega(X,Y))]X(\ln \lambda)\\
& + [-(15-3r-2A(r))\beta(Y) +(10-3r)X (\Omega(X,Y))]Y(\ln \lambda)\\
& + (10-3r)\lambda^2 g(\nabla_X Y + \nabla_Y X, Z)\Omega(X,Y)Z(\ln \lambda)\\
& + X(\beta(X)) - Y(\beta(Y))-\tfrac{2(10-3r)}{r-2}\Omega(X,Y)\varphi^* \Ric^N (X, Y)\\
& + \lambda^2 g([X,Y],Y)\beta(X) +\lambda^2 g([X,Y],X)\beta(Y)\\
&+\lambda^2 \left(g([X,Z],X)-g([Y,Z],Y)\right)\beta(Z)\big\}.
\end{split}
\end{equation} 
Analogously, by propagating the following equation (cf. \eqref{ricphm3})
\begin{equation}\label{hessxxzz}
\begin{split}
& \Hess_{\ln \lambda} (X,X) - A(r) X (\ln \lambda)^2 + \tfrac{1}{r-2}\varphi^* \Ric^N (X, X) + \tfrac{1}{2(r-2)}\lambda^{8-2r}\Omega(X,Y)^2\\
&=\Hess_{\ln \lambda} (Z,Z) - A(r) Z (\ln \lambda)^2 + \tfrac{1}{r-2}\varphi^* \Ric^N (Z, Z),
\end{split}
\end{equation}
we obtain
\begin{equation}\label{vhessxxzz}
\begin{split}
&(3-r)V(\ln \lambda)\big\{(5-r-A(r))[X(\ln \lambda)^2 - Z(\ln \lambda)^2] \\
&+ \tfrac{1}{r-2}[\varphi^* \Ric^N (Z, Z)- \varphi^* \Ric^N (X, X)]\big\} \\
&+\tfrac{3r^2 -20r + 40}{2(r-2)} \tfrac{\lambda^{8-2r}}{4} V(\ln \lambda) \Omega(X,Y)^2 \\
=&-\tfrac{\lambda^{8-2r}}{4}\big\{
(10-3r)(13-3r-A(r))\Omega(X,Y) X (\ln \lambda)Y (\ln \lambda)\\
& + [(15-3r-2A(r))\beta(X) +(10-3r)\lambda^2 g([Z,Y],Z) \Omega(X,Y)]X(\ln \lambda)\\
& + (10-3r)[X(\Omega(X,Y)) - \lambda^2 g([Z,X],Z) \Omega(X,Y)]Y(\ln \lambda)\\
& + [-(15-3r-2A(r))\beta(Z) +(10-3r)\lambda^2 g(\nabla_X Y, Z) \Omega(X,Y))]Z(\ln \lambda)\\
& + X(\beta(X)) - Z(\beta(Z))-\tfrac{10-3r}{r-2}\Omega(X,Y)\varphi^* \Ric^N (X, Y)\\
& -\lambda^2 g([Z,X],Z)\beta(X) +\lambda^2 (g([X,Y],X)-g([Z,Y],Z))\beta(Y)\\
&+\lambda^2 g([X,Z],X)\beta(Z)\big\}.
\end{split}
\end{equation}
Notice that, inside the brackets of both \eqref{vhessxxyy} and \eqref{vhessxxzz}, the polynomial expressions in the derivatives of $\ln \lambda$ have basic coefficients.

\medskip
Now we exploit the commutation of the covariant second order horizontal derivatives. Since $[X,Z]$ is basic, insert \eqref{ricphm2+} in 
$$X(Z(V(\ln \lambda)))- Z(X(V(\ln \lambda)))- [X,Z](V(\ln \lambda))=0,$$
then substitute the $2^{nd}$ order derivatives of $\ln \lambda$ by means of \eqref{hessxy}, to obtain
\begin{equation}\label{c1}
\begin{split}
&\tfrac{(r-4)(10-3r)}{r-2}\Omega(X,Y) Y(\ln \lambda)Z(\ln \lambda)   \\
& + \{ (5-r)\beta(Z) - (10- 3r) \Omega (\nabla_X Z, X)\}X(\ln \lambda)  \\
& - (10-3r)\Omega(X,[Z,Y]) Y(\ln \lambda) \\
& + \{ -(5-r)\beta(X) -(10-3r)\lambda^2 g([Z,Y],Z)\Omega(X, Y)\} Z(\ln \lambda)  \\
& + \dif \beta (X,Z) + \tfrac{10-3r}{r-2}\Omega(X,Y)\varphi^* \Ric (Y,Z)=0,
\end{split}
\end{equation}
whose left-hand term is a polynomial expression in $X(\ln \lambda)$, $Y(\ln \lambda)$ and $Z(\ln \lambda)$ with basic coefficients, denoted for simplicity as follows: \\ $b Y(\ln \lambda) Z(\ln \lambda) + b_1 X(\ln \lambda) + b_2 Y(\ln \lambda) + b_3 Z(\ln \lambda) + b_0$. 

\medskip
Similarly, from $Y(Z(V(\ln \lambda)))- Z(Y(V(\ln \lambda)))= [Y,Z](V(\ln \lambda))$ (or simply by permuting $X$ and $Y$ in the above relation) we obtain
\begin{equation}\label{c2}
\begin{split}
& - \tfrac{(r-4)(10-3r)}{r-2}\Omega(X,Y) X(\ln \lambda)Z(\ln \lambda)  \\
& + (10-3r)\Omega([Z,X],Y) X(\ln \lambda)  \\
& + \{ (5-r)\beta(Z) - (10 - 3r) \Omega (\nabla_Y Z, Y)\}Y(\ln \lambda)   \\
& + \{ -(5-r)\beta(Y) + (10-3r)\lambda^2 g([Z,X],Z)\Omega(X, Y)\} Z(\ln \lambda)  \\
& + \dif \beta (Y,Z) - \tfrac{10-3r}{r-2}\Omega(X,Y)\varphi^* \Ric^N (X,Z)=0,
\end{split}
\end{equation}
whose left-hand term has basic coefficients too and is written down as $- b X(\ln \lambda) Z(\ln \lambda) + c_1 X(\ln \lambda) + c_2 Y(\ln \lambda) + c_3 Z(\ln \lambda) + c_0$.

\medskip
Finally, using the fact that $[X,Y]^\HH$ is basic, insert \eqref{ricphm2+} in 
$$
X(Y(V(\ln \lambda)))- Y(X(V(\ln \lambda)))- [X,Y](V(\ln \lambda))=0,
$$
then replace the $2^{nd}$ order derivatives of $\ln \lambda$ by means of \eqref{tr2} and \eqref{hessxx}, to obtain
\begin{equation}\label{c3}
\begin{split}
&\Omega(X,Y)\{C(r)[X(\ln \lambda)^2 + Y(\ln \lambda)^2]+ D(r)Z(\ln \lambda)^2\}\\
&+d_1 X(\ln \lambda)+d_2 Y(\ln \lambda)+d_3 Z(\ln \lambda)+L(\lambda^{r-2}, \lambda^{8-2r})\\
&+5(4-r)\Omega(X,Y) \lambda^{-2}U(\ln \lambda)^2=0,
\end{split}
\end{equation}
where $C(r)=-\tfrac{2(r-4)(5r^2 -26r +30)}{3(r-2)}$, $D(r)=-\tfrac{(r-4)(2r-5)(5r-18)}{3(r-2)}$,  \\
$L(\lambda^{r-2}, \lambda^{8-2r}) = \dif \beta^{\HH}(X,Y)+ \Omega(X,Y)\big\{\tfrac{7r - 20}{3}\lambda^{r-2} -\tfrac{15r^2 - 58 r + 40}{12(r-2)}\lambda^{8-2r}\Omega(X,Y)^2 + \tfrac{10- 3r}{r-2}[\varphi^* \Ric^N (X,X)+\varphi^* \Ric^N (Y,Y)]-\tfrac{(r-4) (5 r-18)}{6(r-2)}\Scal^N \big\}$, \\ $d_1=(5-r)\beta(Y)-(10-3r)X(\Omega(X,Y))$, $d_2=-(5-r)\beta(X)-(10-3r)Y(\Omega(X,Y))$
and $d_3=-(10-3r)Z\left(\Omega(X,Y)\right)$ (notice that the functions $d_i$'s are basic).

\section{Strategy of the proof}
Given a perfect fluid $(U, p = \frac{r - 3}{3}\lambda^{r}, \rho = \lambda^{r})$ (or equivalently a local $r$-harmonic morphism) on a 4-dimensional spacetime $(M,g)$ satisfying \eqref{couple+}, we aim to prove that at some point where $\lambda\neq 0$ and for $r\neq 0$ we have either $U(\ln \lambda)=0$ (no expansion) or $\Omega(X,Y)=0$ (no rotation).
The starting point is the observation that if the conjecture is true, then $X(\ln \lambda)$, $Y(\ln \lambda)$ and $Z(\ln \lambda)$ have to be basic (possibly all equal to zero). Moreover, a converse result holds (Proposition \ref{basprop}) providing us with an equivalent form of the conjecture which turns out to be more tractable. Indeed, Equations \eqref{c1} and \eqref{c2} allow us to express $X(\ln \lambda)$ and $Y(\ln \lambda)$ as rational functions (with basic coefficients) of $Z(\ln \lambda)$,  provided that $D = b^2 Z(\ln \lambda)^2 - b(c_1 - b_2)Z(\ln \lambda) + (b_1 c_2 - b_2 c_1)$ is not zero, and then to produce a polynomial equation with basic coefficients in one variable, $Z(\ln \lambda)$, which will be constrained to be basic (together with $X(\ln \lambda)$ and $Y(\ln \lambda)$). So the proof will split up into two major cases ($D \neq 0$ and  $D =0$) to be treated with independent methods. Nevertheless they have in common the following basic tools that will help us to conclude in each case.

\begin{lm}\label{pol}
If a function $f$ on $M$ satisfies a polynomial equation $\alpha_n f^n + ... + \alpha_1 f + \alpha_0 =0$ where $\alpha_i$ are all basic functions, then either $f$ is basic or $\alpha_i=0$ for all $i$.
\end{lm}
\begin{proof}
Iterate the derivative along $V$ of the polynomial equation. At each step we have either $V(f)=0$ or an equation of smaller degree is satisfied. If $f$ is not basic, i.e. $V(f)\neq0$, then after $n$ derivations we obtain $\alpha_n=0$. 
\end{proof}

\begin{lm}\label{pollambda}
If $\lambda \neq 0$ is a solution of the equation 
$$
\sum_{i=1}^{n}\eta_i \lambda^{p_i}=0,
$$
where $\eta_i$'s are basic functions and $\eta_{i_0}\neq 0$, then either $V(\ln \lambda)=0$ or it exists $j\neq i_0$ such that $p_{i_0}=p_j$.
\end{lm}
\begin{proof}
Suppose $V(\ln \lambda)\neq 0$. Take $(n-1)$ times the $V$-derivative of the given equation and simplify it by $V(\ln \lambda)$ and $\lambda^{p_j}$ for every $j\neq i_0$. The resulted equation is $\eta_{i_0} \prod_{j\neq i_0}(p_{i_0}-p_j) =0$ and the conclusion follows. 
\end{proof}

Let $\mathcal{R}_0=\{2, 3,4, \tfrac{10}{3}\}$. Recall that the conjecture is true for $r \in \mathcal{R}_0$. 

\begin{pr}\label{basprop}
Let $(U, p = \frac{r - 3}{3}\lambda^{r}, \rho = \lambda^{r})$ be a shear-free perfect fluid coupled with gravity on $(M,g)$ and let $\{X, Y, Z \}$ be horizontal vectors from a preferred orthonormal frame at a point of $M$.

$(i)$ \ If $Z(\ln \lambda)$ is basic and $Z(\ln \lambda)\neq 0$, then either the expansion or the rotation of the fluid vanishes.

$(ii)$ \ If $X(\ln \lambda)$ and $Y(\ln \lambda)$ are basic and $X(\ln \lambda)^2 + Y(\ln \lambda)^2 \neq 0$, then either the expansion or the rotation of the fluid vanishes.
\end{pr}

\begin{proof}
Since $r \in \mathcal{R}_0$ are already settled, we may assume that $r \notin \mathcal{R}_0$.

$(i)$ \ By hypothesis $Z(\ln \lambda)$ is basic, so $Z(V(\ln \lambda)) =0$ due to \eqref{basic}. Equation \eqref{ricphm2+} becomes:
\begin{equation}\label{zbas}
0 = \tfrac{\lambda^{2(4-r)}}{4} \beta(Z) + (3-r)Z(\ln \lambda)V(\ln \lambda).
\end{equation}

Since $Z(\ln \lambda)\neq 0$ (by hypothesis), we have 
\begin{equation}\label{vlbas}
V(\ln \lambda)=-\frac{\lambda^{2(4-r)}}{4}\frac{\beta(Z)}{(3-r)Z(\ln \lambda)}.
\end{equation}
Differentiating this equation along the vector $X$, implies:
$$
X(V(\ln \lambda))=\tfrac{\lambda^{2(4-r)}}{4}((8-2r)f X(\ln \lambda) + X(f))
$$
where the function $f =-\frac{\beta(Z)}{(3-r)Z(\ln \lambda)}$ is basic. Inserting $X(V(\ln \lambda))$ from \eqref{ricphm2+} in the above equality gives us an equation with basic coefficients in $X(\ln \lambda)$ and $Y(\ln \lambda)$:
\begin{equation}\label{zbasx}
(5-r)f X(\ln \lambda) - (10-3r)\Omega(X,Y)Y(\ln \lambda) =  \beta(X) - X(f).
\end{equation}
Analogously we obtain
\begin{equation}\label{zbasy}
(10-3r)\Omega(X,Y)X(\ln \lambda) + (5-r)f Y(\ln \lambda) =  \beta(Y) - Y(f).
\end{equation}
Since $\beta(Z)=-\lambda^2 \Omega(X,Y)g([X,Y],Z)$ the discriminant of the linear system formed by \eqref{zbasx} and \eqref{zbasy} is
$$\Delta = \Omega(X,Y)^2\left((10-3r)^2+ (5-r)^2\frac{\lambda^4 g([X,Y],Z)^2}{(3-r)^2 Z(\ln \lambda)^2}\right).$$

Let us suppose that $\Omega(X,Y) \neq 0$. Then $\Delta\neq 0$ (since $r\neq \tfrac{10}{3}$) and $X(\ln \lambda)$ and $Y(\ln \lambda)$ are basic functions as solutions of a linear system with basic coefficients.  

By using \eqref{vlbas}, the trace constraint \eqref{tr2} becomes
\begin{equation*}
\begin{split}
&\lambda^{8-2r} \left(\frac{r-5}{32(3-r)^2}\frac{\beta(Z)^2}{Z(\ln \lambda)^2} - \frac{3r-5}{24}\Omega(X,Y)^2\right) + \frac{2r-5}{6}\lambda^{r-2} \\
&-\frac{(2r-5)(r-3)}{6}\left(X(\ln \lambda)^2 + Y(\ln \lambda)^2 + Z(\ln \lambda)^2\right)- \frac{r-3}{12}\Scal^N=0,
\end{split}
\end{equation*}
that is a linear equation in $\lambda^{8-2r}$ and $\lambda^{r-2}$  with basic coefficients. Since the first two coefficients cannot cancel simultaneously, by applying Lemma \ref{pollambda} we conclude that $U(\ln \lambda)=0$.

\medskip
$(ii)$ \ Supposing $X(\ln \lambda)\neq 0$, from \eqref{ricphm2+} we have 
\begin{equation}\label{xybas}
V(\ln \lambda)=-\frac{\lambda^{2(4-r)}}{4}\frac{\beta(X)+(10-3r)Y(\ln \lambda)\Omega(X, Y)}{(3-r)X(\ln \lambda)}.
\end{equation}
The function $f= -\tfrac{\beta(X)+(10-3r)Y(\ln \lambda)\Omega(X, Y)}{(3-r) X(\ln \lambda)}$ is basic according to our hypothesis. Suppose $f \neq 0$ (otherwise $V(\ln \lambda)=0$ and the proof ends). Analogously to the previous case $(i)$, by differentiating \eqref{xybas} along the vector $Z$  and inserting $Z(V(\ln \lambda))$ from \eqref{ricphm2+} we obtain
$$
(5-r)f Z(\ln \lambda) =  \beta(Z) - Z(f).
$$
If $r\neq 5$, then $Z(\ln \lambda)$ is basic. Inserting $U(\ln \lambda)=\tfrac{\lambda^{5-r}}{4}f$ in the trace constraint \eqref{tr2}, we obtain again an equation of the form $\lambda^{8-2r} \cdot \text{basic} + \tfrac{2r-5}{6}\lambda^{r-2} + \text{basic}=0$; then, like in the case $(i)$, we obtain $U(\ln \lambda)=0$. 

If $r = 5$, let us suppose that $Z(\ln \lambda)$ is not basic and show that this leads to a contradiction. From \eqref{hessxy} we deduce that $g(\nabla_X Y, Z)=0$, which inserted in \eqref{vhessxxzz} together with \eqref{xybas} gives us: $-f Z(\ln \lambda)^2+\tfrac{3\lambda^{-2}}{16}f \Omega(X,Y)^2 + \beta(Z) Z(\ln \lambda) + \text{basic term}=0$. Derive along $V$ this equation (and reinsert it into the result) to obtain \\
$f^2 Z(\ln \lambda)^2- \beta(Z)f Z(\ln \lambda) + \text{basic term}=0$. According to Lemma \ref{pol} this implies $f=0$, contradiction.
\end{proof}

\section{The case $D \neq 0$}

Recall that $D = b^2 Z(\ln \lambda)^2 - b(c_1 - b_2)Z(\ln \lambda) + (b_1 c_2 - b_2 c_1)$. In this case the key observation is that $X(\ln \lambda)$, $Y(\ln \lambda)$ are rational functions with basic coefficients of $Z(\ln \lambda)\neq 0$. This allows, upon propagation of \eqref{c1} and \eqref{c2}, to obtain a polynomial equation in $Z(\ln \lambda)$ that leads us via Lemma \ref{pol} to the conclusion that $Z(\ln \lambda)$ is basic and the conjecture is true according to Proposition \ref{basprop}. We mention that the proof makes use of the previously known fact \cite{coll} that the conjecture holds in the  special case of aligned vorticity and acceleration, i.e. $X(\ln \lambda)=Y(\ln \lambda)=0$ (see \cite{seno} for a covariant proof).

Assume $\Omega \neq 0$, $D \neq 0$ and $r \notin \mathcal{R}_0$. From \eqref{c1} and \eqref{c2} we obtain 
\begin{equation}\label{xyofz}
\begin{split}
X(\ln \lambda) &= \tfrac{1}{D}\left(b c_3 Z(\ln \lambda)^2
+ (b_2 c_3 - b_3 c_2 + b c_0) Z(\ln \lambda) + b_2 c_0 - b_0 c_2\right),\\
Y(\ln \lambda) &= \tfrac{1}{D}\left(-b b_3 Z(\ln \lambda)^2
+ (b_3 c_1 - b_1 c_3 - b b_0) Z(\ln \lambda) + b_0 c_1 - b_1 c_0\right).
\end{split}
\end{equation}
Notice that we can also assume $Z(\ln \lambda)\neq 0$; otherwise from  \eqref{xyofz} we deduce that $X(\ln \lambda)$, $Y(\ln \lambda)$ are both basic and either Proposition \ref{basprop} applies (if one of them is not zero) or $\gr \ln \lambda \in \VV$ (if they are both zero) and therefore $U$ is irrotational (if $\lambda \neq$ constant), contradicting the assumption $\Omega\neq 0$, or $U(\ln \lambda)= 0$ (if $\lambda =$ constant).

By differentiating Equation \eqref{c1} along $V$, inserting second order derivatives from \eqref{ricphm2+} and simplifying the result by means of \eqref{c1} and \eqref{c2}, we obtain
\begin{equation}\label{propaXZ}
\begin{split}
& (3-r) V (\ln \lambda)\left\{-b_1 X(\ln \lambda) - b_2 Y(\ln \lambda) -b_3 Z(\ln \lambda) -2 b_0\right\}\\
&=\tfrac{\lambda^{8-2r}}{4}\big\{ (10-3r)\Omega(X,Y)(c_1 + b_2)X(\ln \lambda)\\
& \qquad + [(10-3r)\Omega(X,Y)(c_2 - b_1) - b \beta(Z)]Y(\ln \lambda) \\
& \qquad +[(10-3r)\Omega(X,Y) c_3 - b \beta(Y)]Z(\ln \lambda) \\
& \qquad + (10-3r)\Omega(X,Y)c_0 - b_1 \beta(X)- b_2 \beta(Y)-b_3 \beta(Z)\big\},
\end{split}
\end{equation}
while differentiating \eqref{c2} along $V$ gives us
\begin{equation}\label{propaYZ}
\begin{split}
& (3-r) V (\ln \lambda)\left\{-c_1 X(\ln \lambda) -c_2 Y(\ln \lambda) -c_3 Z(\ln \lambda) - 2 c_0 \right\}\\
&=\tfrac{\lambda^{8-2r}}{4}\big\{ [(10-3r)\Omega(X,Y)(c_2 - b_1)+ b \beta(Z)]X(\ln \lambda)\\
& \qquad -(10-3r)\Omega(X,Y)(c_1 + b_2)Y(\ln \lambda) \\
& \qquad -[(10-3r)\Omega(X,Y) b_3 - b \beta(X)]Z(\ln \lambda) \\
& \qquad- (10-3r)\Omega(X,Y)b_0 - c_1 \beta(X) - c_2 \beta(Y) - c_3 \beta(Z)\big\}.
\end{split}
\end{equation}

Eliminating  $V (\ln \lambda)$ and $\lambda^{8-2r}$ from Equations \eqref{propaXZ} and \eqref{propaYZ} and inserting  $X(\ln \lambda)$ and  $Y (\ln \lambda)$ from \eqref{xyofz} gives us a $6^{th}$ degree polynomial equation with basic coefficients in $Z(\ln \lambda)$: $\mathcal{P}\big(Z(\ln \lambda)\big)=0$ (see Appendix for the explicit form). According to Lemma \ref{pol}, either $Z(\ln \lambda)$ is basic or the coefficients of $\mathcal{P}$ are all vanishing. In the former case Proposition \ref{basprop} applies and the conjecture is true, while in the latter we shall obtain a contradiction (except for three values of $r$). 

Let us suppose that $Z(\ln \lambda)$ is not basic and that the coefficients of $\mathcal{P}$ are all zero, in particular the leading one:
\begin{equation}\label{leadi}
b_{3}^{2}+ c_{3}^{2} - \tfrac{r-4}{r-2}\left(b_3 \beta(X) + c_3 \beta(Y)\right)=0.
\end{equation}

\begin{lm} \label{bc}
If $D \neq 0$, $\Omega \neq 0$ and $Z(\ln \lambda)$ is not basic, then $b_2 + c_1 = 0$ and $b_1 - c_2=0$.
\end{lm}

\begin{proof}
Substitute $X(\ln \lambda)$, $Y(\ln \lambda)$ from \eqref{xyofz} and $X(Z(\ln \lambda))$, $Y(Z(\ln \lambda))$ from \eqref{hessxy}, in Equation \eqref{hessxxyy} to obtain a $7^{th}$ degree polynomial equation with basic coefficients in $Z(\ln \lambda)$ whose leading coefficient $b^6(b_2+c_1)$. Since $Z(\ln \lambda)$ is not basic, $b_2+c_1$ must vanish according to Lemma \ref{pol}.  

Analogously, substituting $X(\ln \lambda)$, $Y(\ln \lambda)$ from \eqref{xyofz} in Equation \eqref{hessxy},
then summing with the corresponding relation with $X$ and $Y$ permuted, we obtain again a $7^{th}$ degree polynomial equation in $Z(\ln \lambda)$ with basic coefficients, the leading one being $b^6(b_1 - c_2)$.
\end{proof}

This Lemma allows us to eliminate $b_1$ and $b_2$ in the following computations. According to \eqref{xyofz} and \eqref{propaXZ}, $V (\ln \lambda)=\tfrac{\lambda^{8-2r}}{4}Q$, where $Q=(b^2 \widetilde{c}_3 Z(\ln \lambda)^3+...)/((r-3)b^2 b_3 Z(\ln \lambda)^3+...)$ is a rational function in $Z(\ln \lambda)$ with basic coefficients. Plugging this into \eqref{vhessxxzz} together with $X(\ln \lambda)$ and  $Y (\ln \lambda)$ from \eqref{xyofz} gives us: 
$$
\frac{\lambda^{8-2r}}{4}= \frac{b^6 u(r)\widetilde{c}_3 Z(\ln \lambda)^9 + \dots (\text{lower degree terms})}{b^6 v(r)\Omega(X,Y)^2 \widetilde{c}_3 Z(\ln \lambda)^7 + \dots (\text{lower degree terms})},
$$ 
where $\widetilde{c}_3 := c_{3} - \tfrac{r-4}{r-2}\beta(Y)$, $u(r)=\frac{(r-3)(r-4)(2r-5)}{2(r-2)}$ and $v(r)=\frac{3r^2 -20r + 40}{2(r-2)}$.

Suppose $r \neq 5/2$ to have $u(r)\neq0$ (notice that $v(r)\neq0$). Replacing $V(\ln \lambda)$ and $\tfrac{\lambda^{8-2r}}{4}$ with their (rational) expression in terms of $Z(\ln \lambda)$ in the identity $V\big(\tfrac{\lambda^{8-2r}}{4}\big)=(8-2r)\frac{\lambda^{8-2r}}{4} V(\ln \lambda)$ and using \eqref{ricphm2+} gives a polynomial equation with basic coefficients in $Z(\ln \lambda)$ whose leading term vanishes if and only if $b_3=0$ or $\widetilde{c}_3=0$. A similar argument, starting from \eqref{propaYZ}, imposes $c_3=0$ or $\widetilde{b}_3=0$, where $\widetilde{b}_3 := b_{3} - \tfrac{r-4}{r-2}\beta(X)$. 

Due to \eqref{leadi} we have to consider only the following two sub-cases.

\subsection{Sub-case $b_3=c_3=0$.} Returning to the polynomial equation $\mathcal{P}=0$  mentioned above, we obtain $b_0 \beta(X)+ c_0 \beta(Y)=0$ from the $5^{th}$ degree coefficient term. Inserting it in the $4^{th}$ and $3^{rd}$ degree coefficient terms we obtain respectively
\begin{equation}\label{b0c0}
\begin{split}
& 2(b_0^{2}+c_0^{2})(10-3r)\Omega(X,Y)+3c_2(b_0 \beta(Y) - c_0 \beta(X))=0, \\
& c_1 \big[7(b_0^{2}+c_0^{2})(10-3r)\Omega(X,Y)+10c_2(b_0 \beta(Y) - c_0 \beta(X))\big]=0.
\end{split}
\end{equation}

If $c_1 \neq 0$, the system \eqref{b0c0} implies $b_0=c_0=0$; therefore $X(\ln \lambda)=Y(\ln \lambda)=0$, that is we are in the known case of aligned vorticity and  acceleration where the conjecture is true \cite{coll}, so $V(\ln \lambda)=0$ that contradicts our assumption that $Z(\ln \lambda)$ is not basic.

If $c_1=0$, the coefficient of the $2^{rd}$ degree term gives us
$$
c_2^2 \left[\tfrac{15r-32}{5(r-2)}(b_0^{2}+c_0^{2})(10-3r)\Omega(X,Y)+4c_2(b_0 \beta(Y) - c_0 \beta(X))\right]=0,
$$
where we have used $\beta(Z)=-\tfrac{2}{5(r-4)}c_2$, consequence of the fact established in Lemma \ref{bc}. So either $c_2=0$ (and, from \eqref{b0c0} we obtain again $b_0=c_0=0$, so $X(\ln \lambda)=Y(\ln \lambda)=0$) or $c_2 \neq 0$ and the system formed by the above equation and the first Equation \eqref{b0c0} implies again $b_0=c_0=0$ (so $X(\ln \lambda)=Y(\ln \lambda)=0$) except for $r=\tfrac{16}{5}$.

\subsection{Sub-case $\widetilde{b}_3=\widetilde{c}_3=0$.} In this case the system formed by the cancellation of coefficients of $\mathcal{P}$ can be solved in $c_1$ and $c_2$ and, after a case by case analysis, we obtain (except for $r=\tfrac{12}{5}$)
$$X(\ln \lambda)=\tfrac{\beta(Y)}{(10-3r)\Omega(X,Y)}, \qquad 
Y(\ln \lambda)= - \tfrac{\beta(X)}{(10-3r)\Omega(X,Y)}.$$
Plugging this into \eqref{ricphm2+} we see that either $X(\ln \lambda)$, $Y(\ln \lambda)$ are both zero (the case of aligned vorticity and acceleration) or $V(\ln \lambda)=0$, which contradicts the assumption that $Z(\ln \lambda)$ is not basic.

\medskip
In conclusion, if $D \neq 0$, the conjecture is true for $r \notin \{\tfrac{12}{5}, \tfrac{5}{2}, \tfrac{16}{5}\}$.

\section{The case $D=0$} 

In this case we deal essentially with the situation when the vorticity is orthogonal to the acceleration, $Z(\ln \lambda) = 0$. Since the compatibility conditions \eqref{c1} and \eqref{c2} become trivial, the remaining one, \eqref{c3}, comes into play. By successive propagation of this condition we can eliminate all terms in $X(\ln \lambda)$ and  $Y(\ln \lambda)$ and obtain a quadratic equation in $V(V(\ln \lambda))$ and $V(\ln \lambda)^2$ with coefficients involving only $\lambda$ and basic quantities. This quadratic equation conserves its form upon iterated $V$-derivations, so the terms in $V(\ln \lambda)$ can be eliminated and finally it results an equation for which Lemma \ref{pollambda} applies to conclude that the fluid must be expansion free (once the rotation is non-zero).

Assume that $D=0$ and $r \notin \mathcal{R}_0$. According to Lemma \ref{pol}, either $b=0$ (that is $\Omega = 0$ and the proof ends) or $Z(\ln \lambda)$ is basic. If $Z(\ln \lambda)$ is basic and $Z(\ln \lambda)\neq 0$, Proposition \ref{basprop} applies and we obtain the conclusion. The remaining case is $Z(\ln \lambda) = 0$. 

For the rest of this section, let us suppose moreover that $Z(\ln \lambda) = 0$ and $\Omega(X,Y)\neq 0$. 

From \eqref{ricphm2+} we see that $\beta(Z)=0$ that is $g([X,Y], Z)=0$ (recall that  $\beta(Z)=-\lambda^2g([X,Y], Z)\Omega(X,Y)$). 
Equation \eqref{c1} reduces to 
\begin{equation}\label{c1red}
\begin{split}
\lambda^2 g(\nabla_X Y, Z)X(\ln \lambda) + \lambda^2 g([Z,Y], Y)Y(\ln \lambda) = \tfrac{b_0}{(10-3r)\Omega(X,Y)},
\end{split}
\end{equation}
while Equation \eqref{c2} reduces to 
\begin{equation}\label{c2red}
\begin{split}
\lambda^2 g([Z,X], X)X(\ln \lambda) + \lambda^2 g(\nabla_X Y, Z)Y(\ln \lambda) = -\tfrac{c_0}{(10-3r)\Omega(X,Y)}.
\end{split}
\end{equation}
Notice that the system formed by \eqref{c1red} and \eqref{c2red} is dependent ($D=0$).

Therefore two situations have to be considered: in \eqref{c1red} at least one of the (basic) coefficients of $X(\ln \lambda)$ and $Y(\ln \lambda)$ is not vanishing, or both are zero.

\subsection{Sub-case $g(\nabla_X Y, Z)^2 + g([Z,Y], Y)^2 \neq 0$} 
Suppose $g(\nabla_X Y, Z) \neq 0$ (the other choice is similar). According to \eqref{c1red}, $X(\ln \lambda)$ is a linear function of $Y(\ln \lambda)$ with basic coefficients: $X(\ln \lambda) = \mathfrak{a}Y(\ln \lambda)+ \mathfrak{b}$. Taking the 
$V$-derivative of this relation and employing \eqref{ricphm2+} we obtain: 
$$
V(\ln \lambda) = \tfrac{\lambda^{8-2r}}{4(r-3)\mathfrak{b}} \big\{\beta(\mathfrak{a}X-Y)+(10-3r)\Omega(X,Y)[(\mathfrak{a}^2+1)Y(\ln \lambda) + \mathfrak{a}\mathfrak{b}]\big\},
$$
where we supposed $\mathfrak{b}\neq 0$ (otherwise the conclusion is immediate). Plugging this into Equation \eqref{vhessxxyy} we obtain a third degree polynomial in $Y(\ln \lambda)$ with basic coefficients. According to Lemma \ref{pol}, either $Y(\ln \lambda)$ is basic or the coefficients are all vanishing. If $Y(\ln \lambda)$ is basic then $X(\ln \lambda)$ is basic and we can apply Proposition \ref{basprop} to conclude. If this is not the case, then the leading coefficient $(\mathfrak{a}^4-1)(10-3r)(3-r)(5-r-A(r))\Omega(X,Y)$ must vanish, so $r = 5/2$ (a case not covered by our proof) or $\mathfrak{a}=\pm 1$. But in the last case, the coefficient second degree  term is $\tfrac{8(r-4)}{r-2}(10-3r)(3-r)\mathfrak{b}\Omega(X,Y)$ and it cannot be zero, contradiction.

\subsection{Sub-case $g(\nabla_X Y, Z)=0$ and  $g([Z,Y], Y)=0$} Inserting this in \eqref{c2red} we obtain that either $X(\ln \lambda)$ is basic, or $g([Z,X],X)=0$. 

Firstly, let us suppose that $g([Z,X],X)\neq 0$ (so $X(\ln \lambda)$ is basic). From \eqref{ricphm2+} we have:
\begin{equation}\label{ricphm2+0}
\tfrac{\lambda^{2(4-r)}}{4} \{ \beta(X) + (10- 3r)Y(\ln \lambda)\Omega (X,Y)\} + (3-r)X(\ln \lambda)V(\ln \lambda) =0.
\end{equation}
If $X(\ln \lambda)=0$, then $Y(\ln \lambda)=-\frac{\beta(X)}{(10- 3r)\Omega (X,Y)}$ is basic, so either Proposition \ref{basprop} applies (if $Y(\ln \lambda)\neq0$), or $\gr \ln \lambda \in \VV$ (if $Y(\ln \lambda)=0$). In the latter case $U$ is hypersurface orthogonal (so irrotational, a contradiction) if $\lambda \neq$ constant, or $U(\ln \lambda)=0$ if $\lambda$ is constant.

\noindent If $X(\ln \lambda)\neq 0$, then eliminating $V(\ln \lambda)$ and $\tfrac{\lambda^{8-2r}}{4}$ from \eqref{vhessxxyy} and \eqref{ricphm2+0} we obtain a $3^{rd}$ degree polynomial equation in  $Y(\ln \lambda)$ with basic coefficients whose leading coefficient is
$(10-3r)(3-r)(5-r-A(r))\Omega(X,Y)$. Again Lemma \ref{pol} imposes that either $Y(\ln \lambda)$ is basic (so Proposition \ref{basprop} applies) or the leading coefficient is zero, i.e. $r=5/2$ (excepted case).

\medskip
Secondly, let us suppose that $g([Z,X],X) = 0$. This remaining sub-case can be resumed by the following assumptions:

\medskip
$\bullet$ \ $Z(\ln \lambda)=0$; 

\medskip
$\bullet$ \ $g(\nabla_X Y, Z) = g(\nabla_Y X, Z)=0$ ;

\medskip
$\bullet$ \ $g([Z,X],X)=g([Z,Y],Y)=0$ ($\Rightarrow \ Z(\Omega(X,Y))=0, \ \di Z=0$).

\medskip
In this last sub-case $Z$ itself is shear-free and divergence-free \footnote{We are close to the situation in which there exists a Killing vector parallel to the vorticity, a notorious case for difficulty despite the drastic simplifications it brings about, see \cite{col}.}. Moreover its 3-dimensional orthogonal distribution is integrable (so $Z$ is aligned with a gradient vector) and minimal. It also follows from the above conditions that $\Hess_{\ln \lambda}(Z,X) = \Hess_{\ln \lambda}(Z,Y)=0$, $\varphi^*\Ric^N (Z,X) = \varphi^*\Ric^N (Z,Y)=0$ and $\dif \beta(Z,X) = \dif \beta(Z,Y)=0$. 

Since, in this case, $\Hess_{\ln \lambda}(Z, Z) = - X(\ln \lambda)^2 - Y(\ln \lambda)^2 + \delta_1 X(\ln \lambda) + \delta_2 Y(\ln \lambda) + \lambda^{-2}U(\ln \lambda)^2$, Equation \eqref{hesszz} becomes:
\begin{equation}\label{hesszzsp}
\begin{split}
&(1- B(r))(X(\ln \lambda)^2 + Y(\ln \lambda)^2) - \delta_1 X(\ln \lambda) - \delta_2 Y(\ln \lambda)\\
&- \tfrac{1}{2}\lambda^{-2}U(\ln \lambda)^2 +  \tfrac{1}{6}\lambda^{r-2}
-  \tfrac{3r-14}{24(r-2)}\lambda^{8-2r}\Omega(X, Y)^2 \\
&-  \tfrac{1}{r-2} \left(\varphi^*\Ric^N (Z,Z)+ \tfrac{r-6}{12}\Scal^N\right)=0,
\end{split}
\end{equation}
where we introduced the notations $\delta_1 = \lambda^{2}g([Z,X],Z)$ and $\delta_2 = \lambda^{2}g([Z,Y],Z)$.

Recall that $\beta(X)=Y(\Omega(X,Y))+ \delta_2 \Omega(X,Y)$ and $\beta(Y)=-X(\Omega(X,Y))- \delta_1 \Omega(X,Y)$ (see the proof of Lemma \ref{basic+}). Divide \eqref{c3} by $\Omega(X,Y)$, multiply \eqref{hesszzsp} with $(10-3r)$ and take their sum to obtain:
\begin{equation}\label{xyuzz}
\begin{split}
& \big(C(r)+(10-3r)(1-B(r)\big)\left[X(\ln \lambda)^2+Y(\ln \lambda)^2\right] \\
&+\tfrac{15-4r}{\Omega(X,Y)}\left[\beta(Y)X(\ln \lambda)-\beta(X)Y(\ln \lambda)\right]\\
&+\tfrac{30-7r}{2}\lambda^{-2}U(\ln \lambda)^2 + \tfrac{11r-30}{6}\lambda^{r-2}
+\tfrac{60 + 44 r - 21 r^2}{24(r-2)}\lambda^{8-2r}\Omega(X,Y)^2\\
&+\tfrac{36 + 12 r - 7 r^2}{12(r-2)}\Scal^N - \tfrac{2(10-3r)}{r-2}\varphi^*\Ric^N (Z,Z) +\tfrac{\dif \beta^\HH (X,Y)}{\Omega(X,Y)} =0.
\end{split}
\end{equation}

From \eqref{ricphm2+} we deduce the following propagation equations:
\begin{equation*}
\begin{split}
&V\left(X(\ln \lambda)^2 + Y(\ln \lambda)^2 \right) = 2 \big\{ \tfrac{\lambda^{8-2r}}{4}\left[\beta(X)X(\ln \lambda)+\beta(Y)Y(\ln \lambda)\right]\\
&+(3-r)V(\ln \lambda)\left(X(\ln \lambda)^2 + Y(\ln \lambda)^2\right)\big\},
\end{split}
\end{equation*}
\begin{equation*}
\begin{split}
& V\left(\beta(Y)X(\ln \lambda)-\beta(X)Y(\ln \lambda)\right)\\
&= (10-3r)\Omega(X,Y)\tfrac{\lambda^{8-2r}}{4}\left[\beta(X)X(\ln \lambda)+\beta(Y)Y(\ln \lambda)\right]\\
&+(3-r)V(\ln \lambda)\left(\beta(Y)X(\ln \lambda)-\beta(X)Y(\ln \lambda)\right),
\end{split}
\end{equation*}
\begin{equation*}
\begin{split}
& V\left(\beta(X)X(\ln \lambda)+\beta(Y)Y(\ln \lambda)\right)\\
&=\tfrac{\lambda^{8-2r}}{4}\big\{\beta(X)^2+\beta(Y)^2-(10-3r)\Omega(X,Y)\left[\beta(Y)X(\ln \lambda)-\beta(X)Y(\ln \lambda)\right]\big\}\\
&+(3-r)V(\ln \lambda)\left(\beta(X)X(\ln \lambda)+\beta(Y)Y(\ln \lambda)\right).
\end{split}
\end{equation*}
The "cyclic" property of the $V$-derivatives of the linear term $\beta(Y)X(\ln \lambda)-\beta(X)Y(\ln \lambda)$ allows us to eliminate it by taking the derivative of Equation \eqref{xyuzz} twice along $V$. 

The first $V$-derivative of Equation \eqref{xyuzz} along $V$ (simplified by reinserting $\beta(Y)X(\ln \lambda)-\beta(X)Y(\ln \lambda)$ from \eqref{xyuzz} in the result) reads:
\begin{equation}\label{xyuzzv}
\begin{split}
&\tfrac{2(r-3)(r-4)(11r-30)}{3(r-2)}\tfrac{\lambda^{8-2r}}{4}\left[\beta(X)X(\ln \lambda)+ \beta(Y)Y(\ln \lambda)\right] \\
&+V(\ln \lambda)\big\{\tfrac{2(r-3)^2(r-4)(7r-10)}{3(r-2)}\left[X(\ln \lambda)^2+Y(\ln \lambda)^2\right]\\
&+ \tfrac{2r(2r-5)}{3}\lambda^{r-2} + \tfrac{r(3r-14)(7r-20)}{12(r-2)} \lambda^{8-2r}\Omega(X,Y)^2\\
&-(r-3)\left(\tfrac{8(r-3)}{3(r-2)}\Scal^N + \tfrac{2(10-3r)}{r-2}\varphi^*\Ric^N (Z,Z) - \tfrac{\dif \beta^\HH (X,Y)}{\Omega(X,Y)}\right)\big\} =0.
\end{split}
\end{equation}
Now, taking the derivative along $V$ of Equation \eqref{xyuzzv}, then replacing $\beta(X)X(\ln \lambda)+\beta(Y)Y(\ln \lambda)$ from \eqref{xyuzzv}, $\beta(Y)X(\ln \lambda)-\beta(X)Y(\ln \lambda)$ from \eqref{xyuzz} and also $X(\ln \lambda)^2 + Y(\ln \lambda)^2$ from the trace constraint \eqref{tr2}, gives us an equation of the following form:
\begin{equation}\label{VV}
\begin{split}
&\alpha(r)V(V(\ln \lambda))^2 + \beta(r)V(V(\ln \lambda)) V(\ln \lambda)^2 + \gamma(r)V(\ln \lambda)^4\\
&+V(\ln \lambda)^2 P_1(\lambda^{8-2r}, \lambda^{r-2})
+V(V(\ln \lambda))P_2(\lambda^{8-2r}, \lambda^{r-2})+ P_3(\lambda^{8-2r}, \lambda^{r-2})=0,
\end{split}
\end{equation}
where $P_i$ are polynomial expressions with basic coefficients (see Appendix for the explicit form) and we supposed $r\notin \{\tfrac{5}{2}, \tfrac{30}{11}, \tfrac{15}{4}\}$ to have well defined expressions in $r$.

From \eqref{tr2} we obtain the $3^{rd}$ order $V$-derivative of $\ln \lambda$ in the form:
\begin{equation}\label{VVV}
V(V(V(\ln \lambda))) = V(\ln \lambda) \left(q_1(r) V(V(\ln \lambda)) 
+ q_2(r) V(\ln \lambda)^2 + P_0(\lambda^{8-2r}, \lambda^{r-2}) \right),
\end{equation}
where we have used Equation \eqref{xyuzzv} to eliminate $\beta(X)X(\ln \lambda)+ \beta(Y)Y(\ln \lambda)$ and again \eqref{tr2} to eliminate $X(\ln \lambda)^2 + Y(\ln \lambda)^2$ (see Appendix for the explicit form).
This constraint on $V(V(V(\ln \lambda)))$ implies that Equation \eqref{VV} conserves its general form when it is derived along $V$ (after simplification of $V(\ln \lambda)$, if not zero). More precisely, the propagation of \eqref{VV}, after inserting \eqref{VVV} and dividing by $V(\ln \lambda)$, is
\begin{equation}\label{VV1}
\begin{split}
&\alpha^{(1)}(r)V(V(\ln \lambda))^2 + \beta^{(1)}(r)V(V(\ln \lambda)) V(\ln \lambda)^2 + \gamma^{(1)}(r)V(\ln \lambda)^4\\
&+V(\ln \lambda)^2 P_1^{(1)}(\lambda^{8-2r}, \lambda^{r-2})
+V(V(\ln \lambda))P_2^{(1)}(\lambda^{8-2r}, \lambda^{r-2})
+ P_3^{(1)}(\lambda^{8-2r}, \lambda^{r-2})=0,
\end{split}
\end{equation}
with 
\begin{equation*}
\begin{split}
\alpha^{(1)}(r)=\alpha(r)q_1(r)+2\beta(r), & \quad  P_1^{(1)}=P_1^{\prime}+q_2(r)P_2 - q_1(r)P_1+\beta(r)P_0 , \\
\beta^{(1)}(r)=2\alpha(r)q_2(r)+4\gamma(r), & \quad P_2^{(1)}=P_2^{\prime} + 2P_1+2\alpha(r)P_0 , \\
\gamma^{(1)}(r)=\beta(r)q_2(r)-\gamma(r)q_1(r), & \quad P_3^{(1)}=P_3^{\prime}-q_1(r)P_3 + P_0P_2,
\end{split}
\end{equation*} 
where $P_{i}^{\prime}=\lambda \tfrac{\dif P_{i}}{\dif \lambda}$.

Iterating the $V$-derivative of \eqref{VV} shows that $\alpha^{(k)}$, $\beta^{(k)}$, $\gamma^{(k)}$ form geometric progressions with the same common ratio (or $(\alpha^{(1)}, \beta^{(1)},\gamma^{(1)})$ is an eigenvector of the recurrence matrix). Therefore, from the first four $V$-derivatives of \eqref{VV} we obtain an overdetermined \textit{linear} system in $V(\ln \lambda)^2$ and $V(V(\ln \lambda))$. Its compatibility condition (zero determinant) leads us to a $4^{th}$ degree polynomial equation $\mathcal{P}(\lambda^{8-2r}, \lambda^{r-2})=0$ having basic coefficients:
$$
\eta_1(r)\Omega(X,Y)^2\lambda^{8-2r+3(r-2)} + \eta_2(r)\Omega(X,Y)^8\lambda^{4(8-2r)} 
+ \dots =0, 
$$ 
where $\eta_1$ and $\eta_2$ are rational functions with common zeros $r=0$, $r=\tfrac{10}{3}$ and $r=\tfrac{20}{7}$ (these are common zeros for all $4^{th}$ degree coefficient terms).

Recall that $r\neq 0$ by the conjecture hypothesis and that $r\neq \tfrac{10}{3}$ as we have already assumed. Suppose moreover  $r \neq \tfrac{20}{7}$. If $\eta_1(r)=0$, we employ Lemma \ref{pollambda} with reference coefficient $\eta_2$. Since the corresponding exponent of $\lambda$ is different from the others exponents in $\mathcal{P}$, we conclude that $V(\ln \lambda)=0$. If $\eta_1(r)\neq 0$ we apply Lemma \ref{pollambda} with reference coefficient $\eta_1$ to deduce that either $V(\ln \lambda)=0$ or  $r \in \{ \tfrac{14}{5}, \tfrac{22}{7}, 6\}$. But in the latter case we check that Lemma \ref{pollambda} with reference coefficient $\eta_2$ leads us to $V(\ln \lambda)=0$.

In conclusion, if $D = 0$, the conjecture is true for $r \notin \{\tfrac{5}{2}, \tfrac{30}{11}, \tfrac{20}{7}, \tfrac{15}{4} \}$.

\section{Conclusion}

We have proved the following
\begin{te}
Consider a shear-free perfect fluid solution of Einstein's field equations
where the fluid pressure satisfies a barotropic equation of state of the form $p=w \rho$ with 
$$w \in \RR \setminus \left\{-\tfrac{1}{5}, -\tfrac{1}{6}, -\tfrac{1}{11}, -\tfrac{1}{21}, \tfrac{1}{15}, \tfrac{1}{4}\right\}.$$
Then the fluid is either non-rotating or non-expanding.
\end{te}

The conclusion of the theorem is in fact sharper. On one hand, in the expanding and non-rotating case, the acceleration vanishes (and $\gr^\HH \ln \lambda =0$) and the spacetime must be in this case a warped product of an interval with a constant curvature 3-manifold, that is a FRW model. This was proved in \cite{col83} by analysing the complete list of spacetimes that admit a shear-free non-rotating fluid, classified in terms of the Weyl curvature symmetries. Another, direct proof can be done by starting with compatibility equations as \eqref{vhessxxyy}, \eqref{vhessxxzz} and the propagation of \eqref{hessxy}, and showing that $\gr^\HH \ln \lambda =0$ (this requires a separate analysis for $r \in \{3,4,\tfrac{5}{2} \}$).

On the other hand, in the non-expanding and rotating case, the spacetime must be stationary since the fundamental vector $V$ is Killing.

Thus, adopting the dual perspective of $r$-harmonic morphisms on spacetimes satisfying \eqref{couple+}, the above result states that they are either of warped product type or of Killing type (with six exceptions for $r$).

Concerning the exceptional values of $w$, we notice that the two  positive ones are in the physical regime defined by a speed of sound between 0 and 1, and that all excepted fluids are satisfying the strong energy condition $\rho + 3p >0$. We expect that the proof can be adapted for the remaining values of $w$ as for the already known cases $w \in \{0, \pm \tfrac{1}{3}, \tfrac{1}{9} \}$ that also appeared to be exceptional in our proof. To what extent it can also be adapted for non-linear equations of state is the next step to consider. We postpone these tasks for a future work.

%\bigskip
%\noindent \textbf{Acknowledgements.} I am grateful to Eric Loubeau for many helpful discussions.

\section{Appendix}
\subsection{Notations and conventions}  Throughout the paper $(M, g)$ denotes a connected time-oriented Lorentz 4-dimensional manifold with metric signature $(-,+,+,+)$. We denote by $\mathcal{L}$ the Lie derivative, by $\nabla$ the Levi-Civita connection of $(M,g)$, and we use the following sign conventions for the curvature tensor field $R(X,Y)Z=\nabla_X\nabla_Y Z-\nabla_Y \nabla_X Z-\nabla_{[X,Y]}Z$,  and $\Delta f = \di(\gr f)$ for the Laplacian on functions.  $\Hess_f=\nabla \dif f$ is the Hessian, while $\Ric$ and $\Scal$ denote the Ricci and scalar curvature of $(M, g)$, respectively.

Given a 2-form $\Omega$ on $(M, g)$, the interior product with a vector $X$ is $\imath_X \Omega=\Omega(X, \cdot)$ and its co-differential is defined in analogy with Riemannian case by $\delta \Omega(X) =-\sum_{i=1}^{4}\varepsilon_i(\nabla_{e_i}\Omega)(e_i, X)$, where $\{e_i\}_{i=\ov{1,4}}$ is an orthonormal frame with $g(e_i, e_i)=\varepsilon_i =\pm 1$. The metric $g$ on $M$ induces a (pointwise) metric on the bundle of $p$-covariant tensors (or $p$-forms) on $M$, defined by:
$\langle \mathcal{A}, \mathcal{B}\rangle=\tfrac{1}{p !}\sum_{i_1, ..., i_p=1}^{4} \varepsilon_{i_1} \dots \varepsilon_{i_p}\mathcal{A}(e_{i_1}, \dots, e_{i_p})\mathcal{B}(e_{i_1},\dots, e_{i_p})$, that provides us with the norm $\abs{\mathcal{A}}$ of such object.  The notation $\varphi^* \mathcal{A}$ refers to the usual \textit{pullback} of $\mathcal{A}$ by the mapping $\varphi$.

\medskip
Along the paper we extensively used \texttt{Mathematica} (\cite{math}) to simplify the coefficients and to solve polynomial equations. Many of the formulae involved in the final argument in both principal cases are very long and were not included in the main text. We reproduce here some of them. Full version is available in separate files.

\subsection{Details for the case $D \neq 0$} When $b_1=c_2$ and $b_2=-c_1$ the polynomial we employed is $\mathcal{P}(Z(\ln \lambda))=\sum_{i=0}^{6} \textbf{Coeff}_i \, Z(\ln \lambda)^i$, where 

\medskip
\noindent\(\textbf{Coeff}_6 = -b^4 \big[b (b_3 \beta(X)+c_3 \beta(Y)) - \big(b_3^2+c_3^2\big) \Omega(X,Y) (10-3 r)\)\big],

\medskip
\noindent\(\textbf{Coeff}_5 =-b^3 \big[2 b^2 (b_0 \beta(X)+c_0 \beta(Y))
- 2b ((2 b_3 c_1- c_2 c_3) \beta(X) + (b_3 c_2 + 2 c_1 c_3) \beta(Y) ) + 3\left(c_1 \big(b_3^2+c_3^2\big)  - b (b_0 b_3+c_0 c_3)\right)(10-3r) \Omega(X,Y) \)\big],

\medskip
\noindent\(\textbf{Coeff}_4 = -b^2\big[-(10-3r)\big(3c_1^2 + c_2^2 \big) \big(b_3^2+c_3^2\big) \Omega(X,Y) + b^2 \big(b_3\beta(Z) c_0-b_0 (9 \beta(X) c_1+3 \beta(Y) c_2+\beta(Z) c_3) - 
2b_0^2 \Omega(X,Y) (10-3r)+c_0(-9 \beta(Y) c_1 + 3 \beta(X) c_2 -2(10-3r)c_0 \Omega(X,Y))\big) - b \big(b_3^2 \beta(Z) c_2-2 b_3 c_1 (3 \beta(X) c_1+3 \beta(Y) c_2+5 b_0 \Omega(X,Y)(10-3r))+c_3 \big(-6 \beta(Y) c_1^2+6 \beta(X) c_1 c_2+\beta(Z) c_2 c_3 - 10(10-3r) c_0 c_1 \Omega(X,Y) \big)\big)\)\big],

\medskip
\noindent\(\textbf{Coeff}_3 =b\big[-c_1 \big(c_1^2+c_2^2\big) \big(b_3^2+c_3^2\big) \Omega(X,Y) (10-3r) - 2 b \big(b_3^2 \beta(Z) c_1 c_2-b_3 \big(2 \beta(X) c_1^3+\big(3 c_1^2+c_2^2\big) (\beta(Y) c_2+2 b_0 \Omega(X,Y) (10-3 r))\big)+c_3 \big(-2 \beta(Y) c_1^3 + \beta(X) \big(3 c_1^2 c_2 + c_2^3\big) + \beta(Z) c_1 c_2 c_3-60 c_0 c_1^2 \Omega(X,Y) - 20 c_0 c_2^2 \Omega(X,Y) + 6 c_0 \big(3 c_1^2+c_2^2\big) \Omega(X,Y) r\big)\big)+b^2 \big(2 b_3 \beta(Z) (2 c_0 c_1+b_0 c_2)-2 b_0 \big(8 \beta(X) c_1^2+5 \beta(Y) c_1 c_2+\beta(X) c_2^2+2 \beta(Z) c_1 c_3\big) - 7 b_0^2 c_1 \Omega(X,Y) (10-3r) + c_0 \big(10 \beta(X)
c_1 c_2-2 \beta(Y) \big(8 c_1^2+c_2^2\big)+2 \beta(Z) c_2 c_3 - 7 c_0 c_1 \Omega(X,Y) (10-3r)\big)\big)\)\big],

\medskip
\noindent\(\textbf{Coeff}_2 = b\big[ b^2 \beta(Z) \big(b_0^2+c_0^2\big) c_2+ b \big(-2 b_3 \beta(Z) \big(2 b_0 c_1 c_2+c_0 \big(3 c_1^2+c_2^2\big)\big)+2 b_0
\big(\beta(X) \big(7 c_1^3+3 c_1 c_2^2\big)+\big(3 c_1^2+c_2^2\big) (2 \beta(Y) c_2+\beta(Z) c_3)\big) + 3 b_0^2 \big(3 c_1^2+c_2^2\big) \Omega(X,Y) (10-3r) + c_0 \big(2 \beta(Y) \big(7 c_1^3+3 c_1 c_2^2\big)-4 c_2 \big(\beta(X) \big(3 c_1^2+c_2^2\big)+\beta(Z) c_1 c_3\big) + 3 c_0 \big(3 c_1^2+c_2^2\big)(10-3r) \Omega(X,Y) \big)\big)+\big(c_1^2+c_2^2\big) \big(b_3^2  c_2 \beta(Z) + b_3 \big(\beta(X) \big(-c_1^2+c_2^2\big)- 2 c_1 (\beta(Y) c_2+3 b_0 \Omega(X,Y) (10-3 r))\big)+ c_3 \big(2 \beta(X) c_1 c_2+\beta(Y)\big( -c_1^2+c_2^2 \big)+ \beta(Z)c_2 c_3 - 6 c_0 c_1 \Omega(X,Y) (10-3r)\big)\big)\)\big],

\medskip
\noindent\(\textbf{Coeff}_1=-2 b^2 \big(b_0^2+c_0^2\big) c_1 c_2 \beta(Z) + \big(c_1^2+c_2^2\big)^2(b_0 b_3+c_0 c_3) \Omega(X,Y) (10-3 r)+b \big(c_1^2+c_2^2\big) \big(2 b_3 \beta(Z) (2 c_0 c_1+b_0 c_2)-2 b_0 c_1 (3 \beta(X)c_1+3 \beta(Y) c_2+2 \beta(Z) c_3) - 5 b_0^2 c_1 \Omega(X,Y) (10-3r) + c_0 \big(-6c_1^2 \beta(Y) + 6c_1 c_2
\beta(X)  + 2 c_2 c_3\beta(Z) - 5 c_0 c_1 \Omega(X,Y) (10-3 r)\big)\big)\),

\medskip
\noindent\(\textbf{Coeff}_0=(c_1^2+c_2^2)\big[-b_3 c_0 \big(c_1^2+c_2^2\big)\beta(Z) + b_0 \big(c_1^2+c_2^2\big)(\beta(X) c_1+\beta(Y) c_2+\beta(Z) c_3)+c_0 \big(b \beta(Z) c_0 c_2+\big(c_1^2+c_2^2\big)(\beta(Y) c_1-\beta(X) c_2+c_0 \Omega(X,Y) (10-3 r))\big)+b_0^2 \big(b \beta(Z) c_2+\big(c_1^2+c_2^2\big)\Omega(X,Y) (10-3r)\big)\)\big].

\subsection{Details for the case $D=0$} In Equation \eqref{VV},
\begin{equation*}
\begin{split}
&\alpha(r)=4\tfrac{(3-r)(r-4)(7r-10)}{(r-2)(2r-5)}, \quad 
\beta(r)=\alpha(r)\tfrac{27r^2-257r+510}{2(11r-30)}, \\
&\gamma(r)=\alpha(r)\tfrac{(11-3r)(3r^2+23r-90)}{2(11r-30)},
\end{split}
\end{equation*}
and
\begin{equation*}
\begin{split}
P_1(\lambda^{8-2r}, \lambda^{r-2})=&k_1\lambda^{8-2r} + k_2 \lambda^{6-r} + k_3 \lambda^{16-4r},\\
P_2(\lambda^{8-2r}, \lambda^{r-2})=&\ell_1\lambda^{8-2r} + \ell_2 \lambda^{6-r} + \ell_3 \lambda^{16-4r},\\
P_3(\lambda^{8-2r}, \lambda^{r-2})=&\tfrac{(r-3)(r-4)(11r-30)}{24(r-2)}\Omega(X,Y)^2 \lambda^{16-4r}\left(m_1\lambda^{8-2r} + m_2 \lambda^{6-r} + m_3 \lambda^{16-4r}\right),
\end{split}
\end{equation*}
where

\noindent\(k_1=\frac{(3-r)^2 (7200-5940 r+1760 r^2-257 r^3+21 r^4)}{3(r-2)(2r-5) (11r-30)} \Scal^N - \frac{(r-3)(r-5)(19r-54)}{11r-30}\Lambda_Z ,\)

\noindent\(k_2=\frac{-21600+39840 r-26944 r^2+8588 r^3-1312 r^4+78 r^5}{3(r-2)(11r -30)} ,\)

\noindent\(k_3= \frac{12960000-49248000 r+70419600 r^2-53293200 r^3+24116984 r^4-6806516 r^5+1184369 r^6-117200 r^7+5082 r^8}{24 (r-2)^2 (2 r-5)(4r-15) (11r-30)}\Omega(X,Y)^2 ,\)

\noindent\(\ell_1=-\frac{(3-r)^2 r (7r-22)}{3(2r-5)(r-2)}\Scal^N +(3-r)\Lambda_Z ,\)

\noindent\(\ell_2=\frac{2 (-120+164r -68r^2+9 r^3)}{3 (r-2)} ,\)

\noindent\(\ell_3=\frac{144000 - 422400r + 439960 r^2 - 224068 r^3 + 60574 r^4 - 8353 r^5 + 462 r^6}{24 (r-2)^2 (2r-5)(4r-15)}\Omega(X,Y)^2 ,\)

\noindent\(m_1=-\frac{3 (10-3 r)^2}{2 (2r-5)(4r-15)}\Scal^N + \frac{10-3r}{4r-15}\Lambda_Z + \frac{\beta(X)^2+\beta(Y)^2}{\Omega(X,Y)^2}  , \) 

\noindent\(m_2=-\frac{r(r-6) (10-3 r)^2}{6 (r-2)(r-3)(4r-15)} , \)  

\noindent\(m_3=-\frac{(10-3 r)^2 (30 -23r + 5 r^2)}{4 (r-3) (r-2) (2r -5) (4r-15)} \Omega(X,Y)^2 ,\)

\noindent where $\Lambda_Z:= 2\frac{10-3r}{r-2}\varphi^*\Ric^N (Z,Z) - \frac{\dif \beta^\HH (X,Y)}{\Omega(X,Y)}$.

\medskip
In Equation \eqref{VVV},
$$
q_1(r)=-\tfrac{91r^2-547r+810}{11r-30}, \quad q_2(r)=-\tfrac{(3r-11)(29r^2-168r+240)}{11r-30}
$$
and
\begin{equation*}
\begin{split}
P_0(\lambda^{8-2r}, \lambda^{r-2})=& \left(\tfrac{(3-r)^2 (r-2)(10-3r)}{(r-4)(11r-30)}\Scal^N+\tfrac{(3-r)(r-2)(2r-5)}{2(r-4)(11r-30)}\Lambda_Z \right)\lambda^{8-2r}\\
&-\tfrac{(2r-5)(10-3r)(17r^2-92r+120)}{6(r-4)(11r-30)}\lambda^{6-r}\\
&+\tfrac{r(10-3r)(7r-20)}{8(r-4)(11r-30)}\lambda^{16-4r}\Omega(X,Y)^2\\
\end{split}
\end{equation*}
or, in short notation, $P_0(\lambda^{8-2r}, \lambda^{r-2})=n_1\lambda^{8-2r} + n_2 \lambda^{6-r} + n_3 \lambda^{16-4r}$.

At the level of coefficients, the recurrence relations involved in the iterated $V$-derivatives of Equation \eqref{VV1} read:

\noindent \(\alpha^{(j)}=-K(r)\alpha^{(j-1)} , \beta^{(j)}=-K(r)\beta^{(j-1)}, \gamma^{(j)}=-K(r)\gamma^{(j-1)}, \ j\geq 2\)
where $K(r):=\tfrac{25(r-2)(r-3)}{11r-30}$ and, for $j\geq 1$,

\noindent \(k_{2}^{(j)}=(6-r-q_1(r))k_{2}^{(j-1)}+q_2(r)\ell_{2}^{(j-1)} + \beta^{(j-1)}(r) n_2 , \) 

\noindent \(k_{3}^{(j)}=(16-4r-q_1(r))k_{3}^{(j-1)}+q_2(r)\ell_{3}^{(j-1)} + \beta^{(j-1)}(r) n_3 , \)  

\noindent \(\ell_{2}^{(j)}=(6-r)\ell_{2}^{(j-1)}+2k_{2}^{(j-1)} + 2\alpha^{(j-1)}(r) n_2 , \) 

\noindent \(\ell_{3}^{(j)}=(16-4r)\ell_{3}^{(j-1)}+2k_{3}^{(j-1)} + 2\alpha^{(j-1)}(r) n_3 , \)  

\noindent where $k_{i}^{(0)}=k_i$, $\ell_{i}^{(0)}=\ell_i$, for all $i$, $\alpha^{(0)}=\alpha$ and $\beta^{(0)}=\beta$ (the analogous recurrence relations for $k_{1}^{(j)}$ and $\ell_{1}^{(j)}$ have not been used).

For $P_{3}^{(j)}$, the recurrence relation gets slightly more complicated since 
\begin{equation*}
\begin{split}
P_{3}^{(1)}(\lambda^{8-2r}, \lambda^{r-2})=\lambda^{16-4r}\big[ & m_{1}^{(1)}\lambda^{8-2r} + m_{2}^{(1)} \lambda^{6-r} + m_{3}^{(1)} \lambda^{16-4r}\\
&+ m_{4}^{(1)}\lambda^{2r-4}+ m_{5}^{(1)}\lambda^{r-2}+ m_{6}^{(1)}\big],
\end{split}
\end{equation*}
where we focus on the following coefficients:

\noindent \(m_{2}^{(1)}=(22-5r-q_1(r))\ov{m}_{2}+\ell_{2}n_3+\ell_{3}n_2 , \) 

\noindent \(m_{3}^{(1)}=(32-8r-q_1(r))\ov{m}_{3}+\ell_{3}n_3 , \) 

\noindent \(m_{4}^{(1)}=\ell_{2}n_2 , \) 

\noindent with $\ov{m}_i:=\tfrac{(r-3)(r-4)(11r-30)}{24(r-2)}\Omega(X,Y)^2 m_i$.

Then, for $j\geq 2$ , $P_{3}^{(j)}$ will have the same form and we have

\noindent \(m_{2}^{(j)} = (22-5r-q_1(r))m_{2}^{(j-1)} +\ell_{2}^{(j-1)}n_3+\ell_{3}^{(j-1)}n_2 , \) 

\noindent \(m_{3}^{(j)} = (32-8r-q_1(r))m_{3}^{(j-1)} +\ell_{3}^{(j-1)}n_3 , \) 

\noindent \(m_{4}^{(j)} =  (12-2r - q_1(r)) m_{4}^{(j-1)} + \ell_{2}^{(j-1)} n_2 . \) 

Taking linear combinations of the $V$-derivatives of Equation \eqref{VV1}, we formed the overdetermined linear system in $V(\ln \lambda)^2$ and $V(V(\ln \lambda))$:
\( \left(P_{1}^{(j)}+K(r)P_{1}^{(j-1)}\right)V(\ln \lambda)^2+ \left(P_{2}^{(j)}+K(r)P_{2}^{(j-1)}\right)V(V(\ln \lambda)) + P_{3}^{(j)}+ K(r)P_{3}^{(j-1)}=0, \) 
where $j=2,3,4$. The determinant of its coefficients matrix (that must vanish) is, after dividing by $\lambda^{4(8-2r)}$, a polynomial $\mathcal{P}(\lambda^{8-2r}, \lambda^{r-2})$ with basic coefficients. For instance the coefficient of $\lambda^{4(8-2r)}$ is 
$$
\Omega(X,Y)^8\left\vert 
\begin{array}{ccc}
k_{3}^{(2)}+K(r) k_{3}^{(1)} & \ell_{3}^{(2)}+K(r) \ell_{3}^{(1)} &  m_{3}^{(2)}+K(r) m_{3}^{(1)}   \\
k_{3}^{(3)}+K(r) k_{3}^{(2)} & \ell_{3}^{(3)}+K(r) \ell_{3}^{(2)} & m_{3}^{(3)}+K(r) m_{3}^{(2)} \\
k_{3}^{(4)}+K(r) k_{3}^{(3)} & \ell_{3}^{(4)}+K(r) \ell_{3}^{(3)} & m_{3}^{(4)}+K(r) m_{3}^{(3)} 
\end{array}
\right \vert .
$$

\end{document}